\documentclass[12pt]{article}
\usepackage[T2A]{fontenc}
\usepackage[utf8x]{inputenc}
\usepackage[english]{babel}
\usepackage[a4paper]{geometry}

\usepackage{comment,titling, epigraph, textcomp,eufrak,mathrsfs, authblk}
\usepackage{amssymb,amsmath,amsfonts,amsthm}
\usepackage{mathtools}
\usepackage{multicol, subcaption}
\usepackage{physics}
\usepackage{fancyvrb}
\usepackage{appendix}

\usepackage{graphicx}
\graphicspath{ {images/} }
\usepackage{tikz, tikz-cd}

\usepackage[nottoc]{tocbibind}

\newtheorem{theorem}{Theorem}

\newtheorem{proposition}{Proposition}
\newtheorem{lemma}{Lemma}
\newtheorem{definition}{Definition}

\newcommand{\Set}[2]{ \left\{\, #1 \: \middle| \: #2 \, \right\} }

\newcommand{\Abs}[1]{\left| #1 \right|}

\newcommand{\Norm}[1]{\left\lVert #1 \right\rVert}

\newcommand{\InnerProduct}[2]{\langle #1 , #2 \rangle}

\renewcommand{\Tr}[1]{\operatorname{tr} \left( #1 \right)}

\newcommand{\TrPart}[2]{ \operatorname{tr}_{#1} \left( #2 \right)}

\newcommand{\Span}[1]{ \operatorname{span} #1 }

\newcommand{\Imag}[1]{ \operatorname{Im} #1 }

\newcommand{\Comp}{\mathbb{C}}

\newcommand{\HilH}{\mathcal{H}}
\newcommand{\HilK}{\mathcal{K}}
\newcommand{\HilE}{\mathcal{E}}

\newcommand{\Bound}[1]{\mathfrak{B}\left( #1 \right)}
\newcommand{\BHilHtoK}{\Bound{\HilH,\HilK}}
\newcommand{\BHilH}{\Bound{\HilH}}
\newcommand{\BHilK}{\Bound{\HilK}}
\newcommand{\BHilE}{\Bound{\HilE}}

\newcommand{\TrClass}[1]{\mathfrak{T}\left( #1 \right)}
\newcommand{\THilH}{\TrClass{\HilH}}
\newcommand{\THilK}{\TrClass{\HilK}}
\newcommand{\THilE}{\TrClass{\HilE}}

\newcommand{\States}[1]{\mathfrak{S}\left( #1 \right)}
\newcommand{\SHilH}{\States{\HilH}}

\newcommand{\Effects}[1]{\mathfrak{E}\left( #1 \right)}
\newcommand{\EHilH}{\Effects{\HilH}}

\newcommand{\GGraph}{\mathcal{G}}

\newcommand{\SGraph}{\mathcal{S}}

\title{Properties of operator systems, corresponding to channels}
\date{\today}
\author[1]{V.I. Yashin \footnote{yashin.vi@phystech.edu}}
\affil[1]{Steklov Mathematical Institute of Russian Academy of Sciences, 8 Gubkina St., Moscow, 119991, Russia}
\begin{document}
\maketitle

\begin{abstract}
It was shown in \cite{Duan_2009}, that in finite-dimensional Hilbert spaces each operator system corresponds to some channel, for which this operator system will be an operator graph. This work is devoted to finding necessary and sufficient conditions for this property to hold in infinite-dimensional case.
\end{abstract}

\section{Introduction}
\label{intro}
The theory of zero-error communication via classical channels was first introduced in C.Shannon's renown paper \cite{Shannon_1956}. He considered the problem of constructing zero-error code as a problem of finding maximal independent subset in a \emph{confusability graph} of a channel, and defined zero-error capacity of a channel as a regularisation of an independence number of a graph. The notion of zero-error capacity for quantum channels was first introduced in \cite{Medeiros_de_Assis_Francisco_2004}. The authors of \cite{Duan_Severini_Winter_2010} proposed studying zero-error communication for quantum channels as a "noncommutative" graph-theoretic problem. They introduced a \emph{noncommutative confusability graph} for a noisy quantum channel and considered the problem of finding zero-error capacity for this graph. To find a correcting code for a quantum channel means to find a \emph{quantum anticlique} for this channel \cite{Weaver_2016,Amosov_2018,Amosov_2019,Amosov_Mokeev_1_2018,Amosov_Mokeev_2_2018,Amosov_Mokeev_2019}. This gave a perspective for creation of noncommutative graph theory as a theory of \emph{operator systems} \cite{Choi_Effros_1977,Paulsen_2003}, which are studied as a part of functional analysis \cite{Dosi_2019}. Some notions of classical graph theory were generalized to the quantum case, such as Lov\'asz number \cite{Duan_Severini_Winter_2010}, Ramsey theory \cite{Weaver_2016}, Tur\'an problem \cite{Weaver_2018} and chromatic number \cite{Cameron_et_al_2006}.

Mostly, operator graphs were studied over finite-dimensional Hilbert space. Generalising the theory to infinite dimensions means to come across a number of functional-analytic complexities, because the structure of quantum channels and operator systems in infinite-dimensional case is rather complicated.

In \cite{Duan_2009,Cubitt_Chen_Harrow_2011} it was shown, that in finite-dimensional Hilbert spaces each operator system is a noncommutative confusability graph for some channel. In this work we consider the generalisation of this property for infinite-dimensional separable Hilbert spaces. Differently speaking, we investigate the images of unital completely positive maps in details. In section \ref{notation} we remind basic notions of quantum communication theory. Section \ref{preliminaries} contains known results about finite-dimensional quantum confusability graphs, which will be used in section \ref{main_results}. In section \ref{main_results} we find the necessary and sufficient conditions for an operator system to be a confusability graph of a given quantum channel.

\section{Notation}
\label{notation}
Let us introduce some basic notions. We will denote Hilbert spaces as $ \HilH, \HilK, \HilE $, bounded operators between $\HilH$ and $\HilK$ as $\BHilHtoK$. In case $ \HilH = \HilK $ this space is an algebra, which will be denoted as $\BHilH$. We will write the identity operator in $\BHilH$ as $I$. The space of self-adjoint operators is denoted as $\BHilH_{sa}$, the space of trace-class operators on $\HilH$ is denoted as $\THilH$. The operator $ \rho \in \THilH $ is called a \emph{state}, if it is self-adjoint, positive and $\Tr{\rho} = 1 $. The set of all states is denoted as $\SHilH$. The operator $ E \in \BHilH $ is called an \emph{effect}, if it is self-adjoint and $0 \leq E \leq I$. The set of all effects is denoted as $\EHilH$.

We will consider various topologies on $\BHilH$. By default, $\BHilH$ has a norm topology. Topology on $\BHilH$, generated by seminorms $ \Norm{\cdot}_x : B \mapsto \Norm{B x} $, is called \emph{strong operator topology} and is denoted $so$. Topology on $\BHilH$ generated by seminorms $ \Norm{\cdot}_{x,y} : B \mapsto \Abs{\InnerProduct{x}{By}} $ is called \emph{weak operator topology} and is denoted $wo$. The space $\BHilH$ is dual to $\THilH$. Therefore, there is a $*$-weak topology on $\BHilH$, which will be called \emph{ultraweak} and denoted $\sigma wo$. It is generated by the set of seminorms $ \Norm{\cdot}_\rho : B \rightarrow \Abs{ \Tr{ B \rho } }, \rho \in \THilH $. The closures of $so$ and $wo$ are equal on convex sets. Also, on bounded convex sets the topologies $ so, wo, \sigma wo $ are equal \cite[p.67-71]{Takesaki_2001}.

Let us consider a linear trace-preserving map $ \Phi : \THilH \rightarrow \THilK $. We will call \emph{dual} to $\Phi$ the unit preserving map $ \Phi^* : \BHilK \rightarrow \BHilH $, for which for every $ \rho \in \THilH, B \in \BHilK $ it is true that
\begin{equation}
  \Tr{\rho \Phi^*(B)} =  \Tr{\Phi(\rho) B}
\end{equation}
The map $\Phi$ is called \emph{quantum channel} \cite{Holevo_2012}, if $\Phi^*$ is also \emph{completely positive}, i.e. $ \Phi^* \otimes Id_n $ is a positive map for all natural $n$, where $Id_n : \Bound{\Comp^n} \rightarrow \Bound{\Comp^n}$ is an identity map.

Let $ \Phi : \THilH \rightarrow \THilK $ be a quantum channel. By the Stinespring dilation theorem, there exists a Hilbert space $\HilE$ and an isometry $ V : \HilH \rightarrow \HilK \otimes \HilE $, such that
\begin{equation}
  \Phi(\rho) = \TrPart{\HilE}{ V \rho V^* }
\end{equation}
Then, the dual channel has the form
\begin{equation}
  \Phi^* : \BHilK \rightarrow \BHilH, \quad B \mapsto V^* ( B \otimes \HilE ) V
\end{equation}
The channel $ \hat{\Phi} $ is called \emph{complementary} \cite{Holevo_2007} to $\Phi$ if
\begin{equation}
\hat{\Phi} : \THilH \rightarrow \THilE, \: \rho \mapsto \TrPart{\HilK}{ V \rho V^* }
\end{equation}
The dual to complementary channel has the form
\begin{equation} \label{eq:hatphistar}
  \hat{\Phi}^*: \BHilE \rightarrow \BHilH, \quad B \mapsto V^* (I_\HilK \otimes B) V
\end{equation}

We are ready to define a notion of operator graph for infinite-dimensional case. In order to do it, we use the relation from \cite[Lemma 1]{Duan_Severini_Winter_2010}.
\begin{definition}
  The \emph{operator graph} of a channel $ \Phi : \BHilH \rightarrow \BHilK $ is a space $ \GGraph(\Phi) = \hat{\Phi}^*(\BHilE) $.
\end{definition}
The self-adjoint, including unit, subspace of $\BHilH$ is called \emph{operator system} \cite{Choi_Effros_1977}. Being an image of completely positive unital map, $\GGraph(\Phi)$ is an operator system.

\section{Preliminaries}
\label{preliminaries}

In the paper \cite{Duan_2009} it was shown that the operator system over finite-dimensional Hilbert space is always an operator graph for some channel. In this section we will give such a construction, following the ideas from \cite[Lemma 2]{Duan_2009} and \cite{Shirokov_Shulman_2016}. In section \ref{main_results} this construction will be generalised to infinite-dimensional case.

It is useful to know the correspondence between the Krauss operators of channel $\Phi$ and it's operator graph $\GGraph(\Phi)$.
\begin{proposition}
  If a channel $ \Phi : \THilH \rightarrow \THilK $ has a Krauss representation
  \begin{equation}
    \Phi(\rho) = \sum_{k=1}^d V_k \rho V_k^*
  \end{equation}
  then the operator graph has the form
  \begin{equation}
    \GGraph(\Phi) = \Span{\Set{V_n^* V_m}{1 \leq n,m \leq d}}
  \end{equation}
\end{proposition}
This fact was proven in \cite[Lemma 1]{Duan_Severini_Winter_2010}. For convenience of reader, we will give this proof:
\begin{proof}
  Let $\HilE$ be a $d$-dimensional Hilbert space with orthonormal basis $ \{ \ket{k} \}_{k=1}^d $. We will define the operator
  \begin{equation}
    V : \HilH \rightarrow \HilK \otimes \HilE, \quad  \ket{\varphi} \mapsto \sum_{k=1}^d (V_k \ket{\varphi}) \otimes\ket{k}
  \end{equation}
  Then
  \begin{equation}
    \Phi(\rho) = \TrPart{\HilE}{ V \rho V^* } = \TrPart{\HilE}{ \sum_{n,m} (V_n \rho V_m^*) \otimes \ket{n}\bra{m} }
  \end{equation}
  The complementary channel has the form
  \begin{equation}
    \hat{\Phi}(\rho) = \TrPart{\HilK}{ V \rho V^* } = \sum_{n,m} \Tr{V_n \rho V_m^*}  \ket{n}\bra{m}
  \end{equation}
  \begin{equation}
    \hat{\Phi}^*(B) = \sum_{n,m} \Tr{\ket{n}\bra{m} B} V_m^* V_n
  \end{equation}
  Varying $B$ over $\BHilE$, we see that
  \begin{equation}
    \GGraph(\Phi) = \hat{\Phi}^*(\BHilE) = \Span{\Set{V_n^* V_m}{1 \leq n,m \leq d}}
  \end{equation}
\end{proof}

Let us fix some operator system $\SGraph \subset \BHilH$, where $\HilH$ is finite-dimensional. We want to construct a channel $\Phi$, such that $\SGraph = \GGraph(\Phi)$. We will need the following lemma of Duan:

\begin{proposition}{\cite[Lemma 2]{Duan_2009}} \label{Duans_lemma}
There exists a linear basis $\{A_k\}_{k=1}^d$ in $\SGraph$, such that $ 0 \leq A_k \leq I, \sum_{k=1}^d A_k = I $.
\end{proposition}
\begin{proof}
Let $d = \dim{\SGraph} < +\infty$. Let us choose some linear basis of hermitian elements $\{B_k\}_{k=1}^d$ in $\SGraph$. Then, choose sufficiently small constant $\alpha > 0$, such that all operators $F_k = I + \alpha B_k$ are positive. Now, choose sufficiently small constant $\beta > 0$, such that $A_1 = I - \beta \sum_{k=2}^d F_k$ is positive and let $A_k = \beta F_k$ for $k = 2, \cdots, d$. Then, all $A_k$ are positive and $\sum_{k=1}^d A_k = I$.
\end{proof}

Now, we are ready to give the construction of a corresponding channel. It is similar to the one proposed in \cite{Shirokov_Shulman_2016}.

\begin{proposition} \label{finite_construction}
  Let $\SGraph$ be an operator system in $\BHilH$. There exists a channel $\Phi : \THilH \rightarrow \THilK$, such that $\SGraph = \GGraph(\Phi)$.
\end{proposition}
\begin{proof}
  Let us fix a basis $\{A_k\}_{k=1}^d$ in $\SGraph$ from proposition \ref{Duans_lemma}. Now we choose $d$-dimensional Hilbert space $\HilE$ with orthonormal basis $\{\ket{k}\}_{k=1}^d$ and construct unital completely positive map $\Psi$:
  \begin{equation}
    \Psi : \BHilE \rightarrow \BHilH, \quad B \mapsto \sum_{k=1}^d \bra{k} B \ket{k} A_k
  \end{equation}
  The image of this map is $ \Imag{\Psi} = \Span{\{A_k\}_{k=1}^d} = \SGraph $. So, if $ \Phi = \hat{\Psi}^* $, then $ \SGraph = \GGraph(\Phi) $. Note that $\hat{\Phi} = \Psi^*$ is a \emph{quantum-classical} ($q-c$) channel \cite[\S 6.4]{Holevo_2012}.

  We can also construct channel $\Phi$ and it's Krauss operators explicitly. Let $ \HilK = \HilH^{\oplus d} $ be a direct orthogonal sum of $d$ copies of $\HilH$, and $i_k$ -- isometrical embeddings from $\HilH$ to $\HilK$, which have the form $ h \mapsto (0, \cdots, h ,\cdots, 0) $. Then $ i_k^* i_l = \delta_{kl} I $. Define the operators $ V_k = i_k A_k^{1/2} : \HilH \rightarrow \HilK $, for which holds $ V_k^* V_l = \delta_{kl} A_k $. Now, we construct a channel $\Phi$ with Krauss operators $\{V_k\}_{k=1}^d$:
  \begin{equation}
    \Phi : \THilH \rightarrow \THilK, \quad \rho \mapsto \sum_{k=1}^d V_k \rho V_k^*
  \end{equation}
  Once again, we get $ \GGraph(\Phi) = \Span{ \Set{ V_k^* V_j }{ 1 \leq k,j \leq d} } = \Span{\{A_k\}_{k=1}^d} = \SGraph $.
\end{proof}

\section{Main results}
\label{main_results}

In \cite{Weaver_2015} N.Weaver defined an operator graph to be an operator system, which is closed in weak operator topology. Let us show that this condition holds for our definition of operator graph.

\begin{lemma}
  If the operator system $ \SGraph \subset \BHilH $ is an operator graph of some channel $ \Phi : \THilH \rightarrow \THilK $, then it is closed in weak operator topology $wo$.
\end{lemma}
\begin{proof}
  Let $ \Psi = \hat{\Phi}^* : \BHilE \rightarrow \BHilH $. Then, by definition, $ \SGraph = \Psi(\BHilE) $. Following the formula \eqref{eq:hatphistar}, we will represent $\Psi$ as a composition of maps $ h : B \mapsto I_\HilK \otimes B $ and $ g : B \mapsto V^* B V $, where $ V : \HilH \rightarrow \HilK \otimes \HilE $ is an isometry. The map $h$ is an injective unital $*$-homomorphism, that is why $h(\BHilE)$ is von Neumann algebra, so is closed in $wo$. Let us denote $ \tilde{\HilH} = V \HilH $ -- a Hilbert subspace of $\HilK \otimes \HilE$. Then $ V V^* = P_{\tilde{\HilH}} $ is an orthoprojection on this subspace. The transformation $g$ is represented as
  \begin{equation}
    g(B) = V^* B V = V^* P_{\tilde{\HilH}} B P_{\tilde{\HilH}} V = U^* P_{\tilde{\HilH}}(B)\vert_{\tilde{\HilH}} U
  \end{equation}
  Here $ U = V\vert^{\tilde{\HilH}} : \HilH \rightarrow \tilde{\HilH} $ is a unitary operator, which is a corestriction of $V$, $ P_{\tilde{\HilH}}(\cdot)\vert_{\tilde{\HilH}} $ is an restriction of operators from $\BHilH$ to $\Bound{\tilde{\HilH}}$. Let us denote $ p = P_{\tilde{\HilH}}(\cdot)\vert_{\tilde{\HilH}} : \Bound{\HilK \otimes \HilE} \rightarrow \Bound{\tilde{\HilH}} $ and $ u = U^*(\cdot)U : \Bound{\tilde{\HilH}} \rightarrow \BHilH $. Then $ g = u \circ p $. We will regard $\Bound{\tilde{\HilH}}$ as a subspace of $\Bound{\HilK \otimes \HilE}$. Then it is closed in $wo$, being a von Neumann algebra. The space $(p \circ h)(\BHilE)$ is closed in $wo$, since it is an intersection of closed subspace $h(\BHilE)$ and closed subspace $\Bound{\tilde{\HilH}}$. Finally, since $u$ is a unitary isomorphism, it preserves closure in $wo$. So, the operator system $ \SGraph = \Psi(\BHilE) = (u \circ p \circ h)(\BHilE) $ is closed in $wo$.
\end{proof}

We have found that it is necessary for an operator system to be $wo$-closed in order to correspond to some channel. We will show that in case of separable Hilbert space it is also a sufficient condition. For a given $wo$-closed operator system, we will generalise a construction of a corresponding channel from section \ref{preliminaries}.

Let $\SGraph$ be a $wo$-closed operator system on the separable Hilbert space $\HilH$. The set $\{x_\alpha\}_\alpha$ in topological vector space $E$ is called \emph{total}, if it's linear span is dense in $E$. Similar to proposition \ref{Duans_lemma}, we want to construct the set of effects $\{A_k\}_{k=1}^\infty$, which is total in $(\SGraph,wo)$.
\begin{lemma} \label{Duans_lemma_infinite}
  There exists a total set of operators $\{A_k\}_{k=1}^\infty$ in $ (\SGraph,wo) $, such that $ 0 \leq A_k \leq I, \sum_{k=1}^\infty A_k = I $.
\end{lemma}
\begin{proof}
  The proof generalises steps from \cite[Lemma 2]{Duan_2009} for infinite-dimensional case. Define $ \SGraph_{sa} = \SGraph \cap \BHilH_{sa} $ -- the space of self-adjoint elements from $\SGraph$. The space $\BHilH$ is dual to $\THilH$, and $\sigma wo$ is a weak-$*$ topology on $\BHilH$. By the Banach-Alaoglu theorem, the unit ball in $\BHilH$ is compact in $\sigma wo$ topology. Since $\HilH$ is separable, $\THilH$ is also separable as a Banach space, so the unit ball in $\BHilH$ is metrizable in $\sigma wo$ topology. Let us denote $X$ the intersection of $\SGraph_{sa}$ and a ball with center in $I$ and radius $1/2$:
  \begin{equation}
    X = \Set{ B \in \SGraph_{sa} }{ \Norm{B - I} \leq \frac{1}{2} }
  \end{equation}
  Being an intersection of closed subspace and a metrizable compact, the set $X$ is compact and metrizable. Also, $ \SGraph_{sa} = \Span{X} $. Metrizable compacts are always separable topological spaces, so we can choose in $X$ some countable dense set $ \{ I, \tilde{A}_2, \tilde{A}_3, \cdots \} $. This sequence is total in $\SGraph$, and, if necessary, we can choose linearly independent subset of it. As long as $ \Norm{I - \tilde{A}_k} \leq 1/2 $, then $ \tilde{A}_k \geq 0 $ and $ \Norm{\tilde{A}_k} < 2 $. For $k>1$ define $ A_k = \tilde{A}_k/2^k $, so $ \Norm{A_k} < 1/2^{k-1} $. For $k=1$ define $ A_1 = I - \sum_{k=2}^\infty A_k $. Since $ \Norm{I - A_1} < \sum_{k=2}^\infty \frac{2}{2^k} = 1 $, then $ 0 \leq A_1 \leq I $. So, we have constructed a set $\{A_k\}_{k=1}^\infty$, such that $ 0 \leq A_k \leq I$ and $\sum_{k=1}^\infty A_k = I $, which is total in $ (\SGraph, \sigma wo) $. Weak operator topology $wo$ is weaker than $\sigma wo$, so this set is also dense in $wo$.
\end{proof}

Note that this set may be chosen linearly independent in the algebraic sense, but may not have some good topological properties, such as being Shauder basis. This fact does not harm our needs, since in proposition \ref{finite_construction} the linear independence of $\{A_k\}_{k=1}^{d}$ was never used.

Now we are ready to generalise construction from previous section of a channel, corresponding to $\SGraph$, to separable infinite-dimensional case $\HilH$. This generalisation comes very naturally, without sufficient changes.
\begin{theorem}
  Let $\SGraph$ be a $wo$-closed operator system in $\BHilH$, where $\HilH$ is separable. There exists a channel $\Phi : \THilH \rightarrow \THilK$, such that $\SGraph = \GGraph(\Phi)$.
\end{theorem}
\begin{proof}
  Fix a total in $\SGraph$ set of effects $\{A_k\}_{k=1}^\infty$, for which $\sum_k A_k = I$. We can define unital completely positive map $\Psi$:
  \begin{equation}
    \Psi : \BHilE \rightarrow \BHilH, \quad B \mapsto \sum_{k=1}^\infty \bra{k}B\ket{k} A_k
  \end{equation}
  The sum is understood as the limit of finite sums in weak operator topology. The image of this map is $ \Imag{\Psi} = \overline{ \Span{ \{ A_k \}_{k=1}^\infty } } = \SGraph $. It means, if $ \Phi = \hat{\Psi}^* $, then $ \SGraph = \GGraph(\Phi) $.

  Let us construct $ \Phi $ and it's Krauss operators explicitly. Let $ \HilK = l_2(\HilH) = \HilH^{\oplus \infty} $ be the direct orthogonal sum of countable copies of $\HilH$, and $i_k$ be the isometric embedding from $\HilH$ to $\HilK$ of the form $ h \mapsto (0, \cdots, h ,\cdots) $. Let $ V_k = i_k A_k^{1/2} : \HilH \rightarrow \HilK $. Then $ V_k^* V_l = \delta_{lk} A_k $. We can construct $\Phi$ as a channel with Krauss operators $\{V_k\}_{k=1}^\infty$:
  \begin{equation}
    \Phi : \THilH \rightarrow \THilK, \quad \rho \mapsto \sum_{k=1}^\infty V_k \rho V_k^*
  \end{equation}
  So, we have constructed $\Phi$, such that $ \GGraph(\Phi) = \overline{ \Span{ \Set{ V_k^* V_j }{ 1 \leq k,j < \infty} } } = \SGraph $.
\end{proof}

\section{Conclusion}
In this paper it was found, that an operator graphs for channels, defined as images of complementary channels, are necessarily $wo$-closed operator systems. For an arbitrary operator graph in a separable Hilbert space, the corresponding channel was constructed. The same construction may probably yield for non-separable case. The main problem would be to generalize lemma \ref{Duans_lemma_infinite}, i.e. to choose a total set of effects in operator system, which sums up to identity. Such a set would be extremely non-constructive.

\section*{Acknowledgements}
I am grateful to G.G.Amosov for useful discussion and comments. The work is supported by Russian Science Foundation under the grant no. 19-11-00086 and performed in Steklov Mathematical Institute of Russian Academy of Sciences.

\bibliographystyle{unsrt}      
\bibliography{bibliography}   

\end{document}